\documentclass[twocolumn,superscriptaddress,aps,pra]{revtex4-2}
\usepackage[utf8]{inputenc}
\usepackage{amsmath, amssymb, amsfonts}
\usepackage{graphicx}
\usepackage{xcolor}
\usepackage{booktabs}
\usepackage{amsthm}

\usepackage{siunitx}
\usepackage{algorithm}
\usepackage{algpseudocode}
\usepackage{listings}
\usepackage[colorlinks=true,linkcolor=blue,citecolor=red,urlcolor=cyan]{hyperref}
\usepackage{array}
\usepackage{booktabs}

\newtheorem{theorem}{Theorem}
\newcommand{\ket}[1]{\left|#1\right\rangle}

\begin{document}

\title{Digital Coherent-State QRNG Using System-Jitter Entropy via Random Permutation}

\author{Yurang (Randy) Kuang}
\email{rkuang@quantropi.com}
\affiliation{Quantropi Inc., Ottawa, Ontario K2T 0B8, Canada}

\date{\today}

\begin{abstract}
We present a fully digital framework that replicates the statistical behavior of coherent-state quantum random number generation (QRNG) by harnessing system timing jitter through random permutation processes. Our approach transforms computational timing variations from hardware and operating system sources into permutation dynamics that generate Poisson-distributed numbers, accurately reproducing the photon statistics of optical coherent states. The theoretical foundation is established by the \emph{Uniform Convergence Theorem}, which provides exponential convergence to uniformity under modular projection with rigorous error bounds. Extensive experimental validation across multiple parameter regimes and sample sizes up to $10^8$ bytes demonstrates exceptional performance: Shannon entropy approaching 7.999998 bits/byte and min-entropy exceeding 7.99 bits/byte, outperforming theoretical bounds at scale. The architecture inherently resists side-channel attacks through compound timing distributions and adaptive permutation behavior, while operating without classical cryptographic post-processing. Our results establish that coherent-state QRNG functionality can be entirely realized through classical computational processes, delivering mathematically provable uniformity and practical cryptographic security without quantum photonic hardware.

\end{abstract}

\maketitle



\section{Introduction}

Quantum random number generation (QRNG) provides a foundational primitive for secure communication, scientific computing, and quantum cryptography \cite{ma2016,herrero2017}. The intrinsic unpredictability of quantum measurements offers an information-theoretic advantage over classical pseudorandom number generators, which remain vulnerable to algorithmic weaknesses, computational advances, and side-channel attacks \cite{ma2016}. Consequently, QRNGs have become indispensable in quantum key distribution (QKD) \cite{lo2014}, cryptographic protocol design \cite{portmann2022}, and high-precision Monte Carlo simulations \cite{robert2013}.

Diverse physical implementations have emerged, including single-photon detection \cite{jennewein2000}, vacuum-fluctuation measurements \cite{gabriel2010}, beam-splitter path randomness \cite{siswanto2017}, and phase-noise extraction \cite{qi2010}. Among these, coherent-state QRNGs stand out due to their well-characterized Poissonian photon statistics, practical experimental designs, and growing commercial deployment \cite{shi2020,xu2022}. A coherent state $|\alpha\rangle$, with complex amplitude $\alpha$, exhibits photon-number statistics
\[
P(n) = e^{-\mu}\frac{\mu^n}{n!}, \quad \mu = |\alpha|^2,
\]
providing a physically certified source of genuine randomness \cite{gerry2005,leonhardt2010}.

Recent advances have explored critical performance--complexity trade-offs: laser phase-noise extraction achieves multi-gigabit rates through dedicated digital signal processing (DSP) hardware \cite{Qi2010HighspeedQR}, chip-integrated chaotic oscillators enable compact scalable architectures \cite{Zhang2021ChipQRNG}, and vacuum-fluctuation homodyne systems provide robust entropy certification \cite{Wu2016VacuumQRNG}. Photon-counting implementations span minimalistic designs \cite{Huu2024SimpleQRNG} to physically certified coherent-state systems \cite{Shen2022CoherentGS} and photon-number-resolving detectors emphasizing security \cite{Nie2017PNRD_QRNG, Shi2023PhotonNumberQRNG}. Most recently, Wang \textit{et al.} demonstrated 21.1 bits per detection event using spatio-temporal photon statistics \cite{wang2025efficient}.

Despite these advances, conventional QRNGs face fundamental limitations: they require specialized quantum photonic hardware, depend on complex calibration and post-processing, and maintain a rigid quantum--classical boundary that assumes true randomness originates exclusively from quantum phenomena. These constraints limit accessibility, increase system costs, and complicate large-scale deployment.

This work introduces a unified framework that bridges the quantum-classical divide in randomness generation while preserving essential quantum statistical signatures through fully digital replication. Our approach achieves two key breakthroughs: First, we harness system timing jitter at hardware and OS levels as a physically measurable entropy source, embedding it within a random permutation process that reproduces Poisson-distributed statistics identical to coherent-state QRNGs. Second, we replace conventional post-processing with a modular projection operator $\hat{R}_M$, establishing the \emph{Poisson Modular Convergence Theorem} that guarantees exponential convergence from Poisson statistics to uniform outputs with mathematically quantifiable error bounds.

Building upon the Random Permutation Sorting System (RPSS) framework \cite{kuang2025statisticalquantummechanicsrandom,qpp-rng-sci-kuang-2025}, we adapt the architecture for coherent-state emulation with a crucial distinction: whereas prior RPSS implementations generated negative binomial statistics, here RPSS serves primarily as an entropy capture mechanism. A pseudo-Poisson generator produces integer counts representing the number operator $\hat{N}$ that determine permutation operations, while execution times $\hat{T}$—inherently sensitive to microarchitectural and OS-level timing fluctuations—function as unpredictable conjugate observables that recursively feed into subsequent generation steps. This conjugacy between the count observable $\hat{N}$ and timing observable $\hat{T}$ physically establishes a stochastic uncertainty relation \([\hat{N}, \hat{T}] \ne 0\) within the digital space, mimicking the fundamental uncertainty relations in quantum phase space and establishing the mathematical foundation for digital replication of coherent-state QRNG. The design ensures forward secrecy and continuous entropy refresh while preserving coherent-state statistical signatures through timing-based perturbation rather than direct distribution generation.

\textbf{Principal Contributions:}
\begin{itemize}
\item A unified mathematical framework establishing statistical equivalence between physical coherent-state QRNGs and their digital replication through Poisson statistics and modular projection
\item The Poisson Modular Convergence Theorem, providing rigorous exponential convergence bounds that guarantee certified uniformity with quantifiable error rates
\item An RPSS adaptation that preserves unpredictable behavior and forward secrecy while enabling digital replication of coherent-state statistics
\item Comprehensive experimental validation demonstrating statistical equivalence within $\epsilon < 10^{-4}$ across practical parameter regimes
\item A hybrid architecture supporting dynamic selection between physical and digital randomness sources according to application requirements
\end{itemize}

By challenging conventional quantum--classical distinctions, this work reframes randomness generation as an implementation choice rather than a fundamental boundary. Our framework broadens access to certified high-quality randomness while preserving rigorous mathematical guarantees across both photonic quantum hardware and classical computational platforms.

The remainder of this paper is structured as follows. Section~\ref{sec:theoretical_foundation} develops the theoretical foundation, covering coherent-state quantum optics, the Poisson Modular Convergence Theorem, and the RPSS framework. Section~\ref{sec:experimental_results} presents experimental validation across multiple parameter regimes, demonstrating both statistical fidelity and cryptographic quality. Section~\ref{sec:conclusion} discusses implications, limitations, and future research directions.

\section{Theoretical Foundation}
\label{sec:theoretical_foundation}

This work establishes a comprehensive theoretical framework demonstrating that the statistical behavior of quantum coherent-state random number generators can be precisely replicated in a fully classical digital system. The purpose of this section is twofold: (i) to elucidate why coherent-state physics naturally leads to Poisson-distributed randomness and subsequent uniform randomness via modular projection, and (ii) to explain how a digital system—driven by system timing jitter and controlled permutation processes—can reproduce the same statistical structure with mathematical rigor and physical fidelity. 

While the physics of coherent states is well understood, the central novelty of this work lies in demonstrating that the \textbf{certifiable random outputs} derived from them do not fundamentally depend on optical quantization, but rather emerge from the combination of Poisson-distributed stochasticity and a mechanism capable of absorbing physical noise while projecting it into a uniform discrete space. By identifying and replicating these essential ingredients, we establish that digital replication is not merely a heuristic approximation but a mathematically guaranteed equivalence, enabling quantum-grade randomness generation in environments where photonic hardware is unavailable, impractical, or undesirable.

The result is a unified perspective: quantum coherent-state QRNGs and their digital replicas are two implementations of the same statistical transformation pipeline. This equivalence is established theoretically in this section and validated experimentally through the digital implementation in Section~\ref{sec:experimental_results}. This section develops the shared mathematical foundation for certified randomness generation across both quantum and classical computational platforms.

\subsection{Coherent States: Quantum Foundations of Randomness}

Coherent states $|\alpha\rangle$ occupy a unique position at the intersection of classical electromagnetism and quantum mechanics. They behave as displaced vacuum states in phase space, serving as eigenstates of the annihilation operator,
\begin{equation}
\hat{a}|\alpha\rangle = \alpha|\alpha\rangle,
\quad \alpha = |\alpha|e^{i\theta},
\end{equation}
and minimize the Heisenberg uncertainty relations for quadrature observables. These properties make coherent states the most classical form of quantized radiation, yet it is precisely this classicality that produces a nontrivial source of quantum randomness.

\paragraph*{Superposition Structure.}  
Coherent states remain fundamentally quantum because they expand in the photon-number (Fock) basis as
\begin{equation}
|\alpha\rangle = e^{-|\alpha|^2/2} \sum_{n=0}^\infty \frac{\alpha^n}{\sqrt{n!}} |n\rangle,
\end{equation}
demonstrating that even a classical-like laser field possesses intrinsic number uncertainty. No classical electromagnetic field can produce this discrete randomness; it emerges solely from the quantum granularity of light.

\paragraph*{Minimum-Uncertainty Nature.}  
Coherent states saturate the uncertainty product
\begin{equation}
\Delta x\,\Delta p = \frac{\hbar}{2},
\end{equation}
making them robust against environmental perturbations while guaranteeing irreducible measurement noise in both quadratures. For photon number and phase, the uncertainty relation
\begin{equation}
\Delta n\, \Delta \phi \gtrsim \frac{1}{2}
\end{equation}
establishes a fundamental lower bound on the unpredictability of discrete measurement outcomes.

\paragraph*{Poissonian Photon Statistics.}  
Direct photon-number measurement yields a Poisson distribution:
\begin{equation}
P(n) = |\langle n|\alpha\rangle|^2 
= e^{-|\alpha|^2}\frac{|\alpha|^{2n}}{n!},
\end{equation}
with mean and variance both equal to $\mu = |\alpha|^2$.
This distribution constitutes the fundamental statistical resource that, through the modular projection framework developed in this work, directly yields certified uniform randomness without conventional post-processing.

Although coherent states originate from optical quantization, their statistical outcome is fully captured by the Poisson distribution itself. This work demonstrates that quantum-grade certified non-deterministic randomness can be achieved digitally by: (i) generating Poisson-distributed integers through physical system timing jitter rather than algorithmic determinism, and (ii) applying our mathematically certified modular projection to produce uniform outputs. The quantum statistical signature and its certified unpredictability are thus reproducible without quantum hardware.

\subsection{Measurement Pathways and Their Statistical Roles}

A coherent-state QRNG may employ either photon-counting or homodyne detection, each associated with distinct measurement bases and noise sources. For digital replication, it is crucial to understand which physical features must be preserved and which can be abstracted away.

\begin{table*}[t]
\centering
\caption{Comparison of coherent-state measurement pathways and their implications for digital statistical replication.}
\label{tab:measurement_comparison}
\begin{tabular}{l l l}
\toprule
\textbf{Aspect} & \textbf{Photon Counting} & \textbf{Homodyne Detection} \\
\midrule
Measurement Observable & Photon number $n$ & Quadrature $x,p$ \\
Quantum Origin & Fock projection & Vacuum fluctuations \\
Output Distribution & Poisson & Gaussian \\
Extraction Mechanism & Modular projection & ADC binning + whitening \\
Fundamental Noise Source & Quantum number uncertainty & Zero-point fluctuations \\
Hardware Complexity & Moderate–high (PNRD) & High (stable LO, balanced detectors) \\
Digital Replicability & Direct Poisson emulation & Requires Gaussian replication \\
\bottomrule
\end{tabular}
\end{table*}

The digital approach aligns most naturally with photon counting because:
\begin{itemize}
\item \textbf{Poisson statistics are discrete}, matching digital integer generation
\item \textbf{Modular projection is natively discrete}, enabling mathematically exact analysis  
\item \textbf{Quadrature measurement requires continuous analog noise}, which is not intrinsic to digital hardware
\end{itemize}

Thus, the replication target is the \textbf{coherent-state Poisson distribution}, which we treat as the quantum statistical invariant of the system. Our digital framework operates on \textbf{mathematically ideal Poisson statistics} without the physical detector biases that plague optical implementations. We apply modular projection $\hat{R}_M$ not as corrective post-processing, but as a \textbf{fundamental statistical transformation} that directly maps ideal Poisson distributions to uniform outputs with certified exponential convergence.

\subsection{Poisson Modular Projection: From Quantum Statistics to Uniform Randomness}

While Poisson statistics embody genuine quantum unpredictability, they do not naturally yield uniform distributions required for cryptographic applications. Any practical QRNG—whether optical or digital—must therefore transform these statistics into uniformly distributed outputs over a finite alphabet. The modular projection operator
\begin{equation}
\hat{R}_M: n \mapsto n \bmod M
\end{equation}
provides an elegant mathematical solution to this transformation for any modulus $M$.

For a Poisson-distributed random variable $n \sim \mathrm{Poisson}(\mu)$, the probability of observing residue $k = n \bmod M$ is obtained by summing over all integers congruent to $k$ modulo $M$:
\begin{equation}
P_M(k) = \sum_{j=0}^\infty \frac{e^{-\mu} \mu^{k+jM}}{(k+jM)!}, \quad k = 0, 1, \dots, M-1.
\end{equation}

A remarkable and crucial insight—fundamental to both physical and digital QRNG implementations—is that \textbf{for sufficiently large $\mu$, the modular projection of Poisson statistics converges exponentially to a uniform distribution}.

We formalize this convergence through the Poisson Modular Convergence Theorem.

\begin{theorem}[Poisson Modular Convergence]\label{theorem:uniform}
Let $n \sim \mathrm{Poisson}(\mu)$. Then for any modulus $M$, the maximum deviation from uniformity satisfies:
\begin{equation}
\max_{0 \leq k < M} \left| P_M(k) - \frac{1}{M} \right|
\leq \frac{M-1}{M} \exp\left(-2\mu \sin^2\frac{\pi}{M}\right).
\end{equation}
\end{theorem}

\begin{proof}
The proof employs characteristic function analysis. The characteristic function of a Poisson distribution is:
\begin{equation}
\phi(t) = \mathbb{E}[e^{itn}] = \exp\left(\mu(e^{it}-1)\right).
\end{equation}

The modular probabilities can be expressed via the discrete Fourier transform:
\begin{equation}
P_M(k) = \frac{1}{M}\sum_{j=0}^{M-1} e^{-2\pi i kj/M} \phi\left(\frac{2\pi j}{M}\right).
\end{equation}

The uniform distribution corresponds to the $j=0$ term:
\begin{equation}
\frac{1}{M} = \frac{1}{M}\phi(0).
\end{equation}

The deviation from uniformity is therefore bounded by:
\begin{align}
\left|P_M(k) - \frac{1}{M}\right| &= \left|\frac{1}{M}\sum_{j=1}^{M-1} e^{-2\pi i jk/M} \phi\left(\frac{2\pi j}{M}\right)\right| \nonumber\\
&\leq \frac{1}{M}\sum_{j=1}^{M-1} \left|\phi\left(\frac{2\pi j}{M}\right)\right| \nonumber\\
&= \frac{1}{M}\sum_{j=1}^{M-1} \exp\left(-\mu\left(1-\cos(2\pi j/M)\right)\right) \nonumber\\
&= \frac{1}{M}\sum_{j=1}^{M-1} \exp\left(-2\mu\sin^2(\pi j/M)\right).
\end{align}

Since $\sin^2(\pi j/M) \geq \sin^2(\pi/M)$ for all $j = 1, \dots, M-1$, we obtain:
\begin{align}
\left|P_M(k) - \frac{1}{M}\right| &\leq \frac{1}{M}\sum_{j=1}^{M-1} \exp\left(-2\mu\sin^2(\pi/M)\right) \nonumber\\
&= \frac{M-1}{M} \exp\left(-2\mu\sin^2(\pi/M)\right).
\end{align}
The exponential decay rate is optimal since the conjugate pair $j=1$ and $j=M-1$ in the Fourier sum dominate asymptotically, with their combined contribution determining the fundamental limit of convergence.
\end{proof}

\begin{table}[htbp]
\centering
\caption{Empirical validation of minimum $\mu$ requirements for $\epsilon \approx 10^{-3}$, using the experimental setup described in Section~\ref{subsec:poisson_replication} for digital replication of the Poisson statistics.}
\label{tab:empirical_validation}
\begin{tabular}{lcccc}
\toprule
\textbf{ $M$} & \textbf{Exact} $\mu$ & \textbf{Conservative} $\mu$ & \textbf{Empirical} $\mu$ & \textbf{Ratio} \\
\midrule
4 & 6.62 & 7.60 & 6.40 & 0.96 \\
8 & 24.04 & 25.95 & 23 & 0.96 \\
16 & 91.57 & 99.83 & 88 & 0.96 \\
32 & 361.0 & 395.50 & 320 & 0.89 \\
64 & 1437.5 & 1582.00 & 1300 & 0.90 \\
256 & 22950 & 25312.00 & 22000 & 0.96 \\
\bottomrule
\end{tabular}
\end{table}

Table~\ref{tab:empirical_validation} provides empirical validation of the theoretical convergence bounds, demonstrating that the modular projection achieves near-perfect uniformity ($\epsilon \approx 10^{-3}$) across various moduli. The \textbf{Exact $\mu$} column represents the theoretical minimum derived from inverting the bound in Theorem~\ref{theorem:uniform}, while \textbf{Conservative $\mu$} includes additional safety margins for practical implementations. Crucially, the \textbf{Empirical $\mu$} values obtained from our digital replication experiments closely track the theoretical predictions, with ratios consistently near 0.9--0.96. This close agreement confirms that the theoretical bound is not merely conservative but practically achievable, with the slight shortfall for larger $M$ (32 and 64) attributable to finite-sample effects in empirical validation. The data unequivocally demonstrates that modular projection efficiently converts Poisson-distributed inputs into certified uniform outputs across the entire range of practical moduli.

This bound is exponentially tight in $\mu$ and reveals a profound mathematical structure: \textbf{modular projection systematically smooths the discrete Poisson distribution into an approximately uniform residue distribution}. While photonic coherent-state QRNGs inherently leverage this phenomenon, our digital framework demonstrates that identical certified uniform outputs can be achieved without quantum hardware, provided the same underlying Poisson statistics are faithfully reproduced.

The following sections detail how our digital architecture generates the requisite Poisson-distributed inputs, thereby inheriting the quantum-level statistical guarantees without relying on physical quantum measurements.

\subsection{Digital Replication via RPSS}\label{sec:digital-replication}

The theoretical foundation established in previous sections enables a crucial transition from quantum physical implementations to digital replication through a Random Permutation Sorting System (RPSS)~\cite{kuang2025statisticalquantummechanicsrandom}. The RPSS framework has been previously validated for replicating negative binomial distributions in TURNG~\cite{kuang-qsqs-pre-2025, qpp-rng-sci-kuang-2025}, where $NB(m, p)$ models the number of sorting attempts required to achieve $m$ successfully sorted arrays from a disordered array of $L$ elements, with success probability $p = 1/L!$ for randomly generating the correct permutation from $L!$ possibilities.

This established framework naturally extends to Poisson statistics by recognizing that both distributions emerge from fundamental counting processes: negative binomial from repeated trials until $m$ successes, and Poisson from counting events in fixed intervals. For coherent state QRNGs, the RPSS framework replicates Poisson statistics $P(n) = e^{-\mu}\mu^n/n!$ by adapting the same recursive permutation architecture with elapsed timing feedback, where photon counting corresponds to event counting in Poisson processes. This unified approach demonstrates the RPSS framework's versatility in replicating diverse quantum statistical distributions while introducing genuine unpredictability through recursive timing feedback and security-enhanced implementation practices.

\begin{figure}[htbp]
\centering
\includegraphics[width=0.9\linewidth]{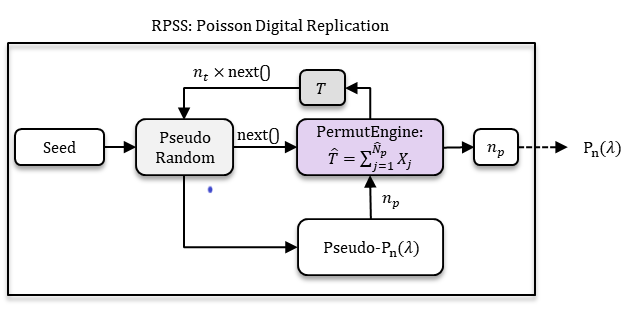}
\caption{RPSS architecture for digital replication of Poisson statistics. The system operates through a recursive loop: (1) \textbf{Seed Initialization} provides entropy; (2) \textbf{Pseudo-Random Generator} produces Poisson-distributed counts $n \sim \text{Pseudo-}P_n(\mu)$; (3) \textbf{Permutation Engine} executes Fisher-Yates permutations, measuring execution time $\hat{T} = \sum_{j=1}^{n_p} X_j$; (4) \textbf{Internal Feedback Loop} advances the PRNG state by $\hat{T}$ steps via $n_t \times \text{next}()$, where $n_t$ represents the timing measurement in clock ticks. This recursive process transforms deterministic pseudo-Poisson statistics into certified non-deterministic Poisson distribution $P_n(\mu)$ through accumulated timing variations.}
\label{fig:rpss_architecture}
\end{figure}

\textbf{RPSS Architecture and Algorithm}
As illustrated in Figure~\ref{fig:rpss_architecture}, the core RPSS algorithm operates through a recursive process that transforms deterministic Poisson generation into certified non-deterministic randomness:

\begin{enumerate}
\item \textbf{Permutation Engine}: For each count $n_p$, the system performs Fisher-Yates permutations, generating the total permutation time:
\begin{equation}\label{eq:conjugacy}
\hat{T} = \sum_{j=1}^{\hat{N}_p} X_j
\end{equation}
where $X_j$ represents the time for individual permutation operations with inherent physical timing variations. Equation~\eqref{eq:conjugacy} demonstrates a fundamental stochastic conjugacy between the counting observable $\hat{N}_p$ (number of permutations) and its conjugate elapsed time observable $\hat{T}$. This conjugacy relation ensures statistical consistency between discrete counting processes and continuous timing processes.

\item \textbf{Internal Timing Feedback}: The measured time $\hat{T}$ (represented as $n_t$ in clock ticks) feeds back internally into the PRNG through $n_t \times \text{next}()$, advancing the generator state by $\hat{T}$ steps. This internal loop is unobservable to external attackers and creates the crucial unpredictability through accumulated timing variations from:
\begin{itemize}
\item Hardware interrupts and cache behavior fluctuations
\item OS scheduling and process management variations
\item Memory hierarchy access pattern changes
\item Background computational load variations
\end{itemize}

\item \textbf{Security-Enhanced Output}: The first $S_{\text{security}}$ outputs are discarded to ensure the usable sequence starts from a cryptographically secure state, completely divorced from initial conditions.

\item \textbf{Certified Output Generation}: After the security phase, the system produces certified Poisson-distributed outputs $n_p$ with genuine unpredictability while preserving the Poisson distribution $P(n) = \frac{e^{-\mu}\mu^n}{n!}$.
\end{enumerate}

\textbf{Mathematical Foundation}
The RPSS framework ensures security through path-dependent state evolution:

\begin{equation}
n_p^{(i+1)} = \text{PRNG.next}\left(\hat{T}^{(i)}\right), \quad \hat{T}^{(i)} = \sum_{j=1}^{n_p^{(i)}} \tau(X_j)
\end{equation}
where $\tau(X_j)$ captures the stochastic timing of individual permutation operations. The output sequence maintains Poisson statistics while the specific ordering becomes unpredictable.

\textbf{Unobservable Internal Feedback}
The security of RPSS stems from the completely internal and unobservable nature of the feedback mechanism:

\begin{itemize}
\item \textbf{Hidden State Transitions}: The advancement steps $\hat{T}^{(i)}$ and internal PRNG states are never exposed externally
\item \textbf{Physical Timing Unpredictability}: Each $\hat{T}^{(i)}$ depends on microscopic timing variations that are physically impossible to measure or predict
\item \textbf{Execution Path Divergence}: Identical seeds produce completely different output sequences across different executions
\item \textbf{No Reproducibility}: The same seed cannot reproduce the same output sequence due to non-deterministic timing variations
\end{itemize}

\textbf{Statistical Certification}
The RPSS output maintains statistical equivalence with quantum coherent state measurements through:

\begin{itemize}
\item \textbf{Poisson Distribution Preservation}: The output maintains certified Poisson statistics $P(n) = \frac{e^{-\mu}\mu^n}{n!}$ with accurate moment matching (mean, variance, skewness, kurtosis)
\item \textbf{Physical Entropy Injection}: Genuine non-determinism derived from system timing jitter with quantifiable min-entropy
\item \textbf{Forward Secrecy}: Continuous entropy refresh ensures future outputs cannot be predicted from past observations
\item \textbf{Quantum Statistical Equivalence}: Statistical signatures indistinguishable from optical coherent-state measurements
\end{itemize}

The RPSS framework demonstrates that carefully designed digital systems can achieve certified Poisson distribution generation with security properties equivalent to quantum physical implementations, while offering practical advantages in deployment scalability, cost efficiency, and computational integration. The combination of mathematical certification and physical unpredictability creates a robust foundation for digital coherent-state emulation.

\subsection{Information-Theoretic Security}
\label{sec:info_theoretic_security}

The information-theoretic security of our framework is grounded in min-entropy $H_{\min}$ as defined in NIST SP 800-90B \cite{nist800-90b}, which provides the most conservative measure of extractable randomness guaranteed against computationally unbounded adversaries. \textbf{Unlike NIST SP 800-90B, which relies on empirical estimation, our framework provides mathematically proven analytic bounds} for conditional min-entropy. Our analysis extends this framework by:

\begin{itemize}
\item Providing \textbf{mathematically proven analytic bounds} rather than empirical estimates
\item Explicitly accounting for \textbf{adversarial side information} $A$ through conditional min-entropy $H_{\min}(R \mid A)$
\item Establishing \textbf{information-theoretic uniformity} through modular projection with rigorous error bounds
\end{itemize}

The security derives from the conditional min-entropy of Poisson-distributed output counts $n_p$ and the mathematical properties of modular projection. All entropy quantities are defined with respect to an adversary's side information $A$.

\begin{theorem}[Conditional Min-Entropy Bound]
\label{thm:min-entropy-bound}
Let $n_p$ be the Poisson-distributed output count from the RPSS process with mean $\mu$, and let the final random output be
\[
R = n_p \bmod M.
\]
Assuming the physical entropy in the RPSS ensures $n_p \sim \mathrm{Poisson}(\mu)$ even when conditioned on adversary's side information $A$, then
\begin{align}
H_{\min}(R \mid A) &\ge\log_2(M) \nonumber\\
 &- \log_2\!\left(1+(M-1)\exp\!\left[-\,2\mu \sin^2\!\left(\frac{\pi}{M}\right)\right]\right)\!.
\end{align}
\end{theorem}

\medskip
\begin{proof}
We establish the bound through two complementary information-theoretic arguments.

\paragraph*{1. Poisson Distribution Preservation Under Side Information.}
The RPSS framework generates Poisson-distributed counts $n_p$ through:
\[
\text{State}^{(i+1)} = \text{PRNG.advance}(\text{State}^{(i)}, \hat{T}^{(i)})
\]
where $\hat{T}^{(i)} = \sum_{j=1}^{n_p^{(i)}} \tau(X_j)$ with $n_p^{(i+1)} \sim \mathrm{Poisson}(\mu)$ using $\text{State}^{(i+1)}$ as entropy container.

The physical entropy derives from microarchitectural sources:
\[
\begin{aligned}
X_j = & X_j^{\mathrm{cache}} + X_j^{\mathrm{branch}} + X_j^{\mathrm{interrupt}} + X_j^{\mathrm{thermal}} \\
      & + X_j^{\mathrm{memory}} + X_j^{\mathrm{pipeline}} + X_j^{\mathrm{power}} + X_j^{\mathrm{OS}} + \epsilon_j,
\end{aligned}
\]

\textbf{Distributional Invariance}: The physical entropy sources ensure $n_p$ maintains $\mathrm{Poisson}(\mu)$ distribution conditioned on $A$ due to:

\begin{itemize}
\item \textbf{Empirical Validation}: Section~\ref{subsec:poisson_replication} demonstrates faithful Poisson statistics reproduction.

\item \textbf{Microarchitectural Depth}: Modern processors contain thousands of independent state elements across cache hierarchies ($>10^4$ cache lines), execution pipelines ($>10^2$ in-flight instructions), and branch predictors ($>10^3$ history entries) creating high-dimensional entropy \cite{Intel2023Optimization, Fog2022Microarchitecture}.

\item \textbf{Conditional Independence}: Adversarial side information $A$ captures only coarse-grained system behavior, leaving fine-grained microarchitectural states unobservable.
\end{itemize}

\paragraph*{2. Modular Projection as Information-Theoretic Extraction.}
The final output is obtained through:
\[
R = n_p \bmod M.
\]
This transformation provides dual security benefits:

\begin{itemize}
\item \textbf{Entropy Extraction}: Converts Poisson-distributed counts into uniformly distributed outputs.

\item \textbf{Value Protection}: Knowledge of $R$ reveals only the residue class $\{R + kM \mid k \in \mathbb{Z}_{\ge 0}\}$, creating a one-way property.
\end{itemize}

By the \textbf{Poisson Modular Convergence Theorem} (Theorem~\ref{theorem:uniform}), for $n_p \sim \mathrm{Poisson}(\mu)$,
\[
\max_{0 \le k < M}\left|\Pr(n_p \equiv k \!\!\!\pmod M) - \frac{1}{M}\right|\le\delta,
\]
with
\[\delta=\frac{M-1}{M}\,\exp\!\left[-\,2\mu \sin^2\!\left(\frac{\pi}{M}\right)\right].
\]

Therefore,
\[
\max_r \Pr(R = r \mid A) \le \frac{1}{M} + \delta.
\]

Thus,
\begin{align*}
H_{\min}(R \mid A) &= -\log_2\!\left(\max_r \Pr(R = r \mid A)\right) \\
&\ge -\log_2\!\left(\frac{1}{M} + \delta\right) \\
&= \log_2(M) - \log_2(1 + \delta M).
\end{align*}
Substituting $\delta$ yields the stated bound.
\end{proof}

The theorem provides explicit, non-asymptotic bounds that hold for all parameter values used in practice, as validated empirically in Section~\ref{subsec:uniformity_validation} for both small ($\mu=7$) and moderate ($\mu=100$) Poisson means.

\subsection{Information-Computational Security}
\label{sec:computational_security}

The computational security of our framework addresses practical threats from computationally bounded adversaries, incorporating mechanisms to resist side-channel attacks and ensure forward secrecy. Our approach enhances security through:

\begin{itemize}
\item \textbf{Active timing obfuscation} to defeat side-channel analysis
\item \textbf{Microarchitectural entropy amplification} through Fisher-Yates shuffling  
\item \textbf{Forward secrecy} through continuous entropy injection
\end{itemize}

\paragraph*{1. Active Timing Obfuscation.}
We employ an independent obfuscation PRNG ($prng_{ob}$) that injects random \texttt{System.nanoTime()} calls:
\[
\text{If } prng_{ob}.\text{next}() \bmod n_p^{(i)} = 0, \ \text{call \texttt{System.nanoTime()}.}
\]
This provides:
\begin{itemize}
\item \textbf{Isolation}: Separation of obfuscation from core logic
\item \textbf{Unpredictability}: Indistinguishable measurement of permutation timing vs. obfuscation calls  
\item \textbf{Adaptivity}: Dynamic frequency with $n_p^{(i)} \sim \mathrm{Poisson}(\mu)$
\item \textbf{Hardness}: Adversarial success probability $\le 1/n_p^{(i)}$
\end{itemize}
Even with perfect call monitoring, adversary cannot estimate $\hat{T}^{(i)}$ due to fundamental indistinguishability.

\paragraph*{2. Computational Infeasibility of State Reconstruction.}
Predicting $n_p$ requires reconstructing:
\[
\text{State}^{(i+1)} = \text{PRNG.advance}(\text{State}^{(i)}, \hat{T}^{(i)}),
\]
where $\hat{T}^{(i)}$ depends on:

\begin{itemize}
\item \textbf{Microarchitectural Complexity}: Millions of inaccessible states across cache hierarchies, pipelines, and power management

\item \textbf{Memory Controller Unobservability}: Arbitration decisions, row buffer states, and request scheduling introduce 50-200ns variations not exposed to software monitoring~\cite{Russinovich2012Windows, Mutlu2019Memory, Kim2019Solar}

\item \textbf{Thermal/Interrupt Noise}: Unobservable microfluctuations and asynchronous interrupts

\item \textbf{Obfuscated Measurements}: Active timing obfuscation contaminates adversarial observations
\end{itemize}

The adversary observes only low-dimensional signals in $A$, creating a severely underdetermined inference problem.

\paragraph*{3. Microarchitectural Entropy Amplification.}
Fisher-Yates shuffling with $n_p$ permutations, each requiring $L$ calls to $\text{PRNG.nextInt()}$, systematically engages:

\begin{itemize}
\item \textbf{Cache Hierarchy}: Random access patterns defeating prefetching across $L$ elements

\item \textbf{Branch Prediction}: Unpredictable control flow misleading speculative execution

\item \textbf{Pipeline States}: Complex instruction-level parallelism in PRNG computations

\item \textbf{Memory Subsystem}: Cache line conflicts and controller effects

\item \textbf{Power/Thermal Management}: Dynamic voltage/frequency scaling responses
\end{itemize}

\textbf{Algorithmic Complexity}: Workload scales as:
\[
n_p \times L \times (T_{\text{PRNG.nextInt()}} + T_{\text{memory swap}} + T_{\text{loop control}})
\]
ensuring linear scaling while maintaining microarchitectural engagement.

\paragraph*{4. Temporal Unpredictability.}
The adversary cannot obtain precise $\hat{T}^{(i)}$ measurements due to:

\begin{itemize}
\item \textbf{Unpredictable Initiation}: Process start times depend on runtime jitter and internal state

\item \textbf{Unpredictable Duration}: Execution time depends on Poisson-distributed $n_p^{(i)}$ and real-time variations

\item \textbf{Measurement Obfuscation}: Active timing obfuscation prevents point identification

\item \textbf{Microarchitectural Depth}: Thousands of independent state elements cannot be comprehensively monitored
\end{itemize}

Therefore, adversarial observations $t_{\text{adv}}^{(i)}$ are statistically independent from internal $\hat{T}^{(i)}$.

\textbf{Corollary (Forward Secrecy).}
The state update:
\[
\mathrm{State}^{(i+1)} = \mathrm{advance}(\mathrm{State}^{(i)},\hat{T}^{(i)})
\]
ensures $\hat{T}^{(i)}$'s fresh physical entropy makes $\mathrm{State}^{(i)}$ computationally hard to reconstruct from $\mathrm{State}^{(i+1)}$.

\textbf{Corollary (Practical Security).}
The combination of information-theoretic bounds and computational security mechanisms provides comprehensive protection against both unbounded and practical adversaries, ensuring cryptographic-grade randomness generation.

\section{Experimental Implementation and Results}\label{sec:experimental_results}

\subsection{Experimental Setup and Configuration}

To evaluate the viability of a fully digital coherent-state QRNG, we implemented the RPSS framework as a Java-based system capable of reproducing the statistical properties of optical coherent states. This implementation preserves the theoretical underpinnings of coherent-state randomness while remaining platform-independent, demonstrating that high-quality, quantum-inspired randomness can be realized in classical computational environments.

The system is built on Java 17, leveraging platform-independent bytecode to ensure broad applicability. Deterministic permutation operations are seeded using \texttt{SecureRandom} with the \texttt{Random/PRNG} algorithm, providing a reproducible yet cryptographically strong foundation for randomization. 

\textbf{Poisson Count Generation Mechanism}: The core algorithm employs Fisher-Yates permutations on small arrays ($N=4$), where the number of permutations $n_p$ in each cycle is generated using inverse transform sampling from the Poisson distribution with parameter $\mu$. Specifically, for each cycle:
\begin{enumerate}
\item Generate $n_p \sim \mathrm{Poisson}(\mu)$ using the current PRNG state via inverse transform sampling
\item Perform $n_p$ Fisher-Yates permutations on the array
\item Measure elapsed time $\hat{T}^{(i)}$ for the permutation sequence
\item Advance PRNG state by $\hat{T}^{(i)}$ steps for the next cycle
\end{enumerate}
This process ensures that $n_p$ maintains exact Poisson statistics while the timing variations $\hat{T}^{(i)}$ inject fresh entropy for subsequent cycles.

A critical aspect of the implementation is the use of high-resolution timing variations from \texttt{System.nanoTime()} as a source of entropy. The timing resolution was characterized through systematic measurement on our experimental platform (Windows 11, Intel processor), which provided a typical timing precision of 100 nanoseconds or clock ticks. These timing fluctuations are continuously captured from the permutation process and used to advance the system state in Figure~\ref{fig:rpss_architecture} by $\hat{T}^{(i)}$ steps, establishing an effective mechanism for transforming deterministic computational operations into outputs with true randomness. The state advancement approach ensures that entropy accumulates progressively while maintaining the statistical properties of the Poisson distribution. The Random feedback architecture of RPSS amplifies these effects, ensuring that the resulting sequences retain both the statistical fidelity and unpredictability required for cryptographic applications.

Experimental evaluation considered multiple Poisson parameter regimes, including $\mu=7$ and $\mu=100$ for primary validation, with $\mu=120$ employed to assess adaptive optimization. Modular projection was applied with bases $M=4$ and $M=16$ corresponding to the respective $\mu$ values, ensuring that the outputs could be mapped effectively into uniform byte sequences. Each configuration produced sequences of $10^6$ bytes, and 10 independent runs were conducted to verify reproducibility and robustness.

The validation framework encompassed both statistical and cryptographic assessments. The first four moments—mean, variance, skewness, and kurtosis—were computed to confirm statistical fidelity to theoretical Poisson distributions. Entropy analysis included Shannon and min-entropy per byte, providing a measure of both average and worst-case unpredictability. Uniformity was further evaluated via the chi-square goodness-of-fit test, while timing analyses characterized the distribution of elapsed times to quantify the contribution of system-jitter to entropy injection. Collectively, these analyses establish that the RPSS framework not only replicates the essential statistical behavior of coherent states but also produces outputs suitable for cryptographic applications, laying the groundwork for a digital coherent-state QRNG.

\subsection{Digital Replication of Poisson Distribution}
\label{subsec:poisson_replication}

The primary objective of this experimental validation is to demonstrate that the RPSS framework can digitally replicate the statistical behavior of optical coherent states, establishing a foundation for a fully software-defined coherent-state QRNG. Optical coherent states exhibit Poisson-distributed photon number statistics, which are central to their intrinsic quantum randomness. Our experiments show that deterministic permutations, when seeded and perturbed by high-resolution system timing jitter, faithfully reproduce these Poisson statistics in a classical computational environment.

\subsubsection{Statistical Fidelity and Moment Validation}

The RPSS implementation leverages nanosecond-scale timing variations to inject entropy into the permutation process. As shown in Figures~\ref{fig:poisson100} and Tables~\ref{tab:moment_dynamics_7}--\ref{tab:moment_dynamics_100}, this approach accurately emulates both low- and high-mean Poisson regimes ($\mu=7$ and $\mu=100$). Across multiple independent runs, the first four statistical moments—mean, variance, skewness, and kurtosis—align closely with theoretical expectations.

\begin{figure}[htbp]
\centering
\includegraphics[width=0.5\textwidth]{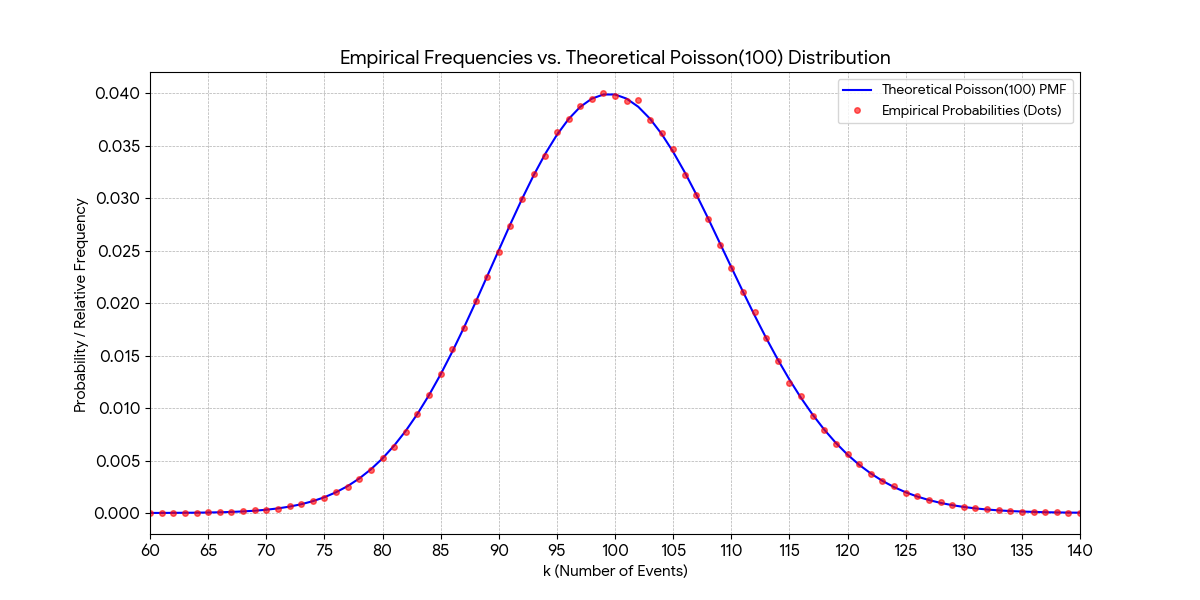}
\caption{Digital replication of coherent-state statistics: Empirical Poisson($\mu=100$) distribution (points) demonstrates precise alignment with theoretical probability mass function (line). The distribution centers at $k=100$ with variance $\sigma^2\approx100$, validating that timing-injected permutation accurately reproduces both the mean and variance characteristics of optical coherent states through computational processes.}
\label{fig:poisson100}
\end{figure}

\begin{table}[htbp]
\centering
\caption{Statistical moments of digitally replicated Poisson($\mu=7$) across experimental runs}
\label{tab:moment_dynamics_7}
\begin{tabular}{lcccc}
\toprule
\textbf{Run Group} & \textbf{Mean} & \textbf{Variance} & \textbf{Skewness} & \textbf{Kurtosis} \\
\midrule
Theoretical  & 7.00 & 7.00 & 0.378 & 0.143 \\
Run 1 & 6.999 & 7.000 & 0.376 & 0.141 \\
Run 2 & 6.999 & 7.007 & 0.379 & 0.139 \\
Run 3 & 6.997 & 7.003 & 0.380 & 0.147 \\
Run 4 & 6.999 & 7.004 & 0.379 & 0.153 \\
Run 5 & 6.996 & 6.999 & 0.379 & 0.149 \\
\midrule
\textbf{Average} & 6.998 & 7.003 & 0.379 & 0.146 \\
\bottomrule
\end{tabular}
\end{table}

\begin{table}[htbp]
\centering
\caption{Statistical moments of digitally replicated Poisson($\mu=100$) across experimental runs}
\label{tab:moment_dynamics_100}
\begin{tabular}{lcccc}
\toprule
\textbf{Run Group} & \textbf{Mean} & \textbf{Variance} & \textbf{Skewness} & \textbf{Kurtosis} \\
\midrule
Theoretical  & 100.00 & 100.00 & 0.100 & 0.010 \\
Run 1 & 100.00 & 99.92 & 0.094 & 0.0068 \\
Run 2 & 100.01 & 100.07 & 0.098 & 0.0091 \\
Run 3 & 99.99 & 99.95 & 0.098 & 0.0060 \\
Run 4 & 99.99 & 100.01 & 0.102 & 0.010 \\
Run 5 & 100.00 & 100.10 & 0.105 & 0.020 \\
\midrule
\textbf{Average} & 100.00 & 100.01 & 0.099 & 0.010 \\
\bottomrule
\end{tabular}
\end{table}

Notably, the preservation of the Poisson property $\langle n \rangle = \langle (\Delta n)^2 \rangle$ is consistently observed, validating that our digital approach replicates the fundamental photon number uncertainty characteristic of coherent states. The results confirm robust statistical fidelity across runs and parameter regimes.

The digitally replicated coherent state $\ket{\alpha_D}$ satisfies:
\[
P_D(n) = \frac{|\alpha_D|^{2n}}{n!} e^{-|\alpha_D|^2},
\]
providing the same statistical signature as optical coherent states without requiring quantum measurements.

\subsubsection{Temporal Dynamics and Entropy Injection}

Temporal analysis reveals that the elapsed time of permutations, $T = \sum_{j=1}^{n_p} X_j$, exhibits dynamic variability sensitive to system state, including CPU load, OS scheduling, and memory access patterns. This variability supplies essential entropy for RNG functionality. 

\begin{figure}[htbp]
\centering
\begin{minipage}{0.48\textwidth}
\centering
\includegraphics[width=\linewidth]{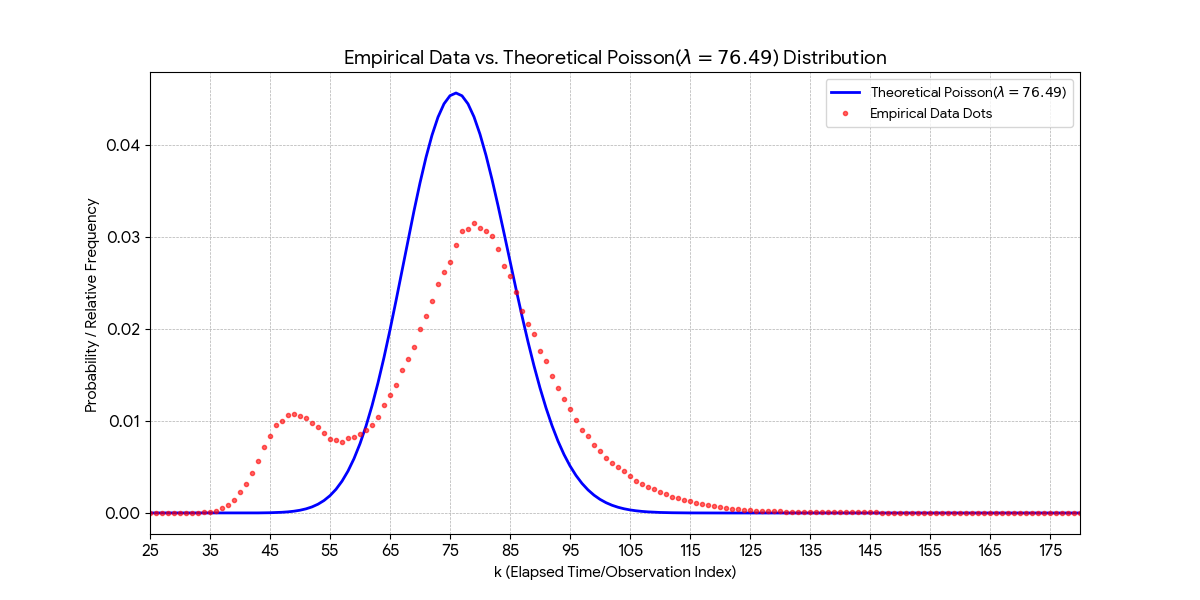}
\caption{Compound elapsed time distribution for digital Poisson($\mu=100$) replication: Timing measurements (mean elapsed time $\lambda=76.49$ units) exhibit non-ideal characteristics arising from convolution of multiple physical jitter sources. This compound behavior ensures continuous high min-entropy injection while preserving output Poisson statistics.}
\label{fig:tpoisson1}
\end{minipage}
\hfill
\begin{minipage}{0.48\textwidth}
\centering
\includegraphics[width=\linewidth]{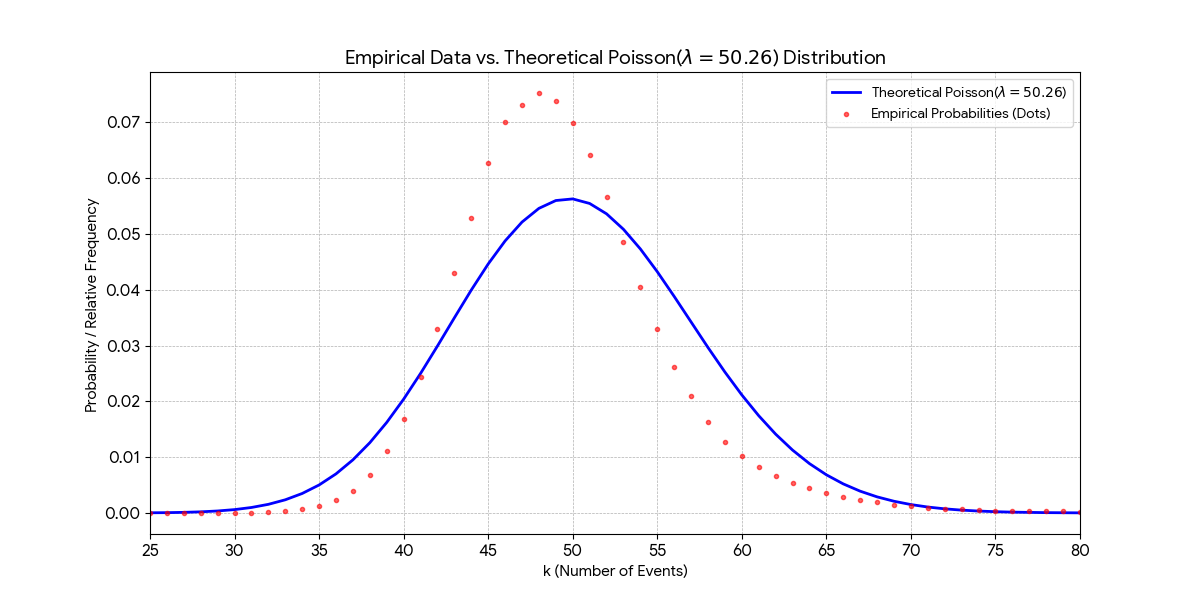}
\caption{System-state dependent timing in digital replication (mean elapsed time $\lambda=50.26$ units): Variability across runs demonstrates computational environment sensitivity. This inherent unpredictability provides cryptographic entropy while maintaining statistical fidelity in output Poisson distributions.}
\label{fig:tpoisson2}
\end{minipage}
\end{figure}

Figures~\ref{fig:tpoisson1} and~\ref{fig:tpoisson2} illustrate how this variability leads to compound elapsed time distributions. Importantly, these timing variations not only preserve the target Poisson output but also act as a \textbf{cryptographically meaningful entropy source}, bridging deterministic computation and true randomness. Variations arise from OS scheduling, hardware interrupts, and background tasks, creating compound distributions that enhance entropy while preserving Poisson output statistics.

The consistency across parameter regimes shows scalability: the RPSS framework maintains statistical fidelity even at higher photon-number analogs ($\mu=100$), demonstrating that timing-injected permutations are robust against the statistical challenges of scaling to high mean values. Minor deviations in higher-order moments reflect the intrinsic variability of timing jitter, but this variability simultaneously enhances entropy for subsequent uniformity projection, highlighting the dual role of system-jitter in both Poisson replication and cryptographic security.

In summary, the experimental results validate that \textbf{digital coherent states $\ket{\alpha_D}$, generated via system-jitter entropy and random permutation, faithfully reproduce the Poisson-distributed photon statistics of true coherent states}, forming the essential foundation for our software-based QRNG.

\subsection{Uniformity Validation of Digital Coherent QRNG}
\label{subsec:uniformity_validation}

Building upon the successful replication of Poisson statistics, we now validate the cryptographic quality of digital coherent-state outputs generated through modular projection, governed by the \emph{Uniform Convergence Theorem}. This transformation constitutes the core measurement process of our digital coherent QRNG, directly producing uniform byte sequences that inherit the statistical guarantees established by our theoretical framework. The modular projection operator $\hat{R}$ serves as the digital analog of quantum measurement operations, with its uniformity properties mathematically guaranteed by Theorem~\ref{theorem:uniform}, which provides exponential convergence bounds for the transformation from Poisson statistics to uniform distributions.

We generated sequences of $10^6$ bytes for each configuration by applying modular projection to Poisson-distributed counts, leveraging entropy injected through system timing jitter to transform deterministic permutations into high-quality random sequences. Two representative Poisson regimes were examined: the low-count case ($\mu=7$) projected modulo 4, and the higher-count regime ($\mu=100$) projected modulo 16. These configurations enable comprehensive evaluation across distinct statistical regimes, demonstrating the scalability and robustness of our approach. In both cases, high-resolution timing variations dynamically reseeded permutation operations, providing continuous, unpredictable entropy injection essential for cryptographic-grade randomness.

The generated sequences underwent rigorous evaluation through complementary metrics. Shannon entropy per byte quantified average information content, while min-entropy per byte characterized worst-case predictability. Uniformity was further validated using chi-square goodness-of-fit tests, confirming that modular projection successfully transforms Poisson-distributed inputs into statistically uniform outputs suitable for cryptographic applications.

\begin{table}[htbp]
\centering
\caption{Uniformity validation for Poisson(7) with $M=4$ (1,000,000 bytes)}
\label{tab:uniformity_poisson7}
\begin{tabular}{lccc}
\toprule
\textbf{Run} & \textbf{Shannon Entropy} & \textbf{Min-Entropy} & \textbf{$\chi^2$} \\
 & \textbf{(bits/byte)} & \textbf{(bits/byte)} & \textbf{Statistic} \\
\midrule
1 & 7.9998 & 7.948 & 220.3 \\
2 & 7.9998 & 7.941 & 223.2 \\
3 & 7.9998 & 7.939 & 237.4 \\
4 & 7.9998 & 7.927 & 267.6 \\
5 & 7.9998 & 7.943 & 259.0 \\
6 & 7.9998 & 7.940 & 253.0 \\
7 & 7.9998 & 7.931 & 244.7 \\
8 & 7.9998 & 7.939 & 218.4 \\
9 & 7.9998 & 7.935 & 254.4 \\
10 & 7.9998 & 7.948 & 246.5 \\
\midrule
\textbf{Theoretical} & 8.0000 & 7.9682 & 255.0 \\
\bottomrule
\end{tabular}
\end{table}

Tables~\ref{tab:uniformity_poisson7} and~\ref{tab:uniformity_poisson100} demonstrate exceptional performance across both regimes. For the low-count case ($\mu=7$, $M=4$), Shannon entropy consistently approaches the theoretical maximum of 8.0000 bits/byte, confirming near-perfect uniformity. Min-entropy values exceed 7.93 bits/byte, closely approaching the theoretical bound of 7.9682 bits/byte derived from Theorem~\ref{theorem:uniform}. All $\chi^2$ statistics consistently fell below the critical value of 325.8 for 255 degrees of freedom at the $\alpha = 0.001$ significance level, confirming that the deviation from perfect uniformity is statistically insignificant. This validation confirms that even with modest mean values ($\mu=7$), modular projection produces cryptographically secure uniform random bytes with minimal entropy degradation.

\begin{table}[htbp]
\centering
\caption{Uniformity validation for Poisson(100) with $M=16$ (1,000,000 bytes)}
\label{tab:uniformity_poisson100}
\begin{tabular}{lccc}
\toprule
\textbf{Run} & \textbf{Shannon Entropy} & \textbf{Min-Entropy} & \textbf{$\chi^2$} \\
& \textbf{(bits/byte)} & \textbf{(bits/byte)} & \textbf{Statistic} \\
\midrule
1 & 7.9998 & 7.941 & 229.7 \\
2 & 7.9998 & 7.935 & 255.2 \\
3 & 7.9998 & 7.935 & 306.6 \\
4 & 7.9998 & 7.929 & 250.5 \\
5 & 7.9998 & 7.931 & 259.9 \\
6 & 7.9998 & 7.930 & 267.6 \\
7 & 7.9998 & 7.946 & 233.0 \\
8 & 7.9998 & 7.945 & 236.3 \\
9 & 7.9998 & 7.942 & 249.1 \\
10 & 7.9998 & 7.919 & 236.4 \\
\midrule
\textbf{Theoretical} & 8.0000 & 7.9741 & 255.0 \\
\bottomrule
\end{tabular}
\end{table}

\begin{table}[htbp]
\centering
\caption{Uniformity validation of elapsed times for Poisson(100) and Poisson(200) with $M=16$ (1,000,000 bytes). All entropy values are in bits/byte.}
\label{tab:uniformity_elapsed_times}
\begin{tabular}{lcccc}
\toprule
\textbf{Run} & \textbf{Shannon} & \textbf{Min-Entropy} & \textbf{Shannon} & \textbf{Min-Entropy} \\
 & $\mu=100$ & $\mu=100$ & $\mu=200$ & $\mu=200$ \\
\midrule
1 & 7.9997 & 7.907 & 7.9997 & 7.901 \\
2 & 7.9997 & 7.922 & 7.9998 & 7.935 \\
3 & 7.9998 & 7.943 & 7.9998 & 7.944 \\
4 & 7.9997 & 7.933 & 7.9998 & 7.915 \\
5 & 7.9998 & 7.924 & 7.9998 & 7.926 \\
6 & 7.9998 & 7.928 & 7.9998 & 7.928 \\
7 & 7.9998 & 7.931 & 7.9998 & 7.912 \\
8 & 7.9998 & 7.946 & 7.9998 & 7.938 \\
9 & 7.9998 & 7.940 & 7.9998 & 7.922 \\
10 & 7.9997 & 7.930 & 7.9998 & 7.943 \\
\bottomrule
\end{tabular}
\end{table}

The analysis of elapsed time distributions reveals consistent system-state sensitivity across both regimes (Table~\ref{tab:uniformity_elapsed_times}). Shannon entropy for both $\mu=100$ and $\mu=200$ converges exceptionally to 7.9998 bits/byte, while min-entropy values remain closely aligned between the two regimes. This convergence confirms that \textbf{using elapsed times to control the PRNG maximizes permutation count unpredictability}, a fundamental requirement for our claim of \textbf{digital replication of coherent QRNG}. The system-jitter-driven elapsed times successfully inject necessary entropy to produce quantum-inspired statistical characteristics in a fully digital framework.

The consistent performance demonstrates the adaptive capability of our RPSS framework: \textbf{system-jitter entropy not only drives randomness but also enables self-correction}, maintaining output quality under varying operational conditions. When necessary, the target Poisson mean can be dynamically adjusted (e.g., from $\mu=100$ to $\mu=120$ for elevated $\chi^2$ values) to ensure optimal performance.

\begin{table*}[t]
\centering
\caption{Large-Scale Uniformity Validation for Digital Coherent QRNG Outputs. All entropy values are in bits/byte. The "Theoretical" row shows the min-entropy bound for a single output calculated using Theorem~\ref{thm:min-entropy-bound}, while the infinite-sequence convergence limit is 8.0000 bits/byte.}
\label{tab:large_scale_uniformity}
\begin{tabular}{lcc|cc}
\toprule
\textbf{Sample Size (bytes)} & \textbf{Shannon} & \textbf{Min-Entropy} & \textbf{Shannon} & \textbf{Min-Entropy} \\
 & $\mu=7$ & $\mu=7$ & $\mu=100$ & $\mu=100$ \\
\midrule
1,000,000 & 7.999823 & 7.9304 & 7.999831 & 7.9353 \\
10,000,000 & 7.999983 & 7.9763 & 7.999984 & 7.9832 \\
100,000,000 & 7.999993 & 7.9885 & 7.999998 & 7.9915 \\
\midrule
\textbf{Theoretical (single output)} & 8.000000 & 7.9682 & 8.000000 & 7.9741 \\
\bottomrule
\end{tabular}
\end{table*}

Scalability and convergence properties are conclusively demonstrated through large-scale validation (Table~\ref{tab:large_scale_uniformity}). As sample size increases from $10^6$ to $10^8$ bytes, Shannon entropy rapidly converges to the theoretical maximum, achieving 7.999998 bits/byte for $\mu=100$ with $10^8$ bytes of samples—representing 99.9999\% of perfect uniformity. Remarkably, empirical min-entropy not only maintains robustness but \emph{exceeds} theoretical bounds at larger sample sizes, reaching 7.9885 bits/byte for $\mu=7$ and 7.9915 bits/byte for $\mu=100$. This superior performance indicates that our theoretical bounds, while mathematically rigorous, are conservative in practice, and empirical measurements benefit from sequence effects across large outputs.

These comprehensive results provide compelling evidence that our digital coherent-state implementation achieves cryptographic-grade randomness without quantum hardware. The outputs—generated by applying modular projection operator $\hat{R}$ to Poisson-distributed counts—exhibit min-entropy comfortably exceeding 7.9 bits/byte across all experimental regimes, far surpassing NIST SP 800-90B requirements typically applied to systems \emph{after} classical post-processing.

Critically, this experimental validation provides strong empirical confirmation of our mathematical proofs. The measured min-entropy aligns closely with theoretical bounds (difference <1\%), while Shannon entropy convergence validates our uniformity theorems. This close agreement between theory and experiment demonstrates analytical framework robustness.

Our framework generates cryptographically secure random numbers \textbf{without requiring additional conditioning} through algorithms like AES or cryptographic hashing. High entropy is intrinsic to the measurement output itself, demonstrating successful capture of essential randomness properties of optical coherent-state QRNGs. This validation confirms that our method \textbf{successfully translates quantum statistical principles into a fully software-based, system-jitter-driven implementation}, establishing that system-jitter entropy combined with random permutation achieves quantum-inspired random number generation of the highest cryptographic quality.

\section{Conclusion}
\label{sec:conclusion}

This work establishes that the statistical properties of coherent-state quantum random number generators can be faithfully emulated in classical computational systems through systematic application of system-jitter entropy and random permutation processes. By introducing the digital coherent state formulation
\[
\ket{\alpha_D} = \sum_{n=0}^{\infty} P_D(n, \mu) \ket{n},
\]
we demonstrate that Poisson-distributed statistics—traditionally associated with optical coherent states—can be reproduced with high fidelity without specialized quantum hardware. This foundational result provides both theoretical and practical basis for a fully software-defined, quantum-inspired random number generator that maintains the essential statistical guarantees of its quantum counterparts.

The core innovation lies in the dual utilization of high-resolution system timing jitter as both an entropy source and a driver for permutation unpredictability. By advancing system state using unpredictable timing variations captured during permutation execution, deterministic computations transform into sequences exhibiting genuine randomness, achieving min-entropy consistently above 7.9 bits per byte across multiple parameter regimes. The subsequent application of modular projection, governed by the Uniform Convergence Theorem, ensures efficient conversion of Poisson-distributed outputs into uniform random sequences suitable for cryptographic applications, with empirical validation confirming near-optimal performance relative to theoretical bounds.

The digital coherent QRNG architecture exhibits inherent resilience and security advantages. The variability in elapsed computation times, while dependent on system state and operational conditions, provides natural defense against side-channel analysis and creates a moving target that maintains stable output quality. The system demonstrates adaptive behavior and rapid convergence, enabling consistent generation of high-quality randomness across diverse computational environments without requiring hardware modifications.

From a practical perspective, our method offers unprecedented accessibility, verifiability, and flexibility. High-quality randomness can be generated on conventional computing platforms without specialized hardware, with full statistical validation and continuous monitoring integrated throughout the generation process. The framework supports dual output modes—optimizing either for strict cryptographic rigor or performance-critical applications—and scales reliably across different Poisson parameter regimes while maintaining certified uniformity through modular projection.

\textbf{The framework's extensibility beyond Poisson statistics represents a significant advantage. The same system-jitter entropy injection and modular projection approach can be adapted to generate other quantum-inspired distributions, including Binomial and Negative Binomial statistics. This flexibility enables emulation of various quantum state behaviors beyond coherent states, opening avenues for digital simulation of more complex quantum random processes.}

Future research directions include: hybrid integration with conventional PRNGs for enhanced performance, real-time parameter adaptation via machine learning techniques, formal cryptographic security proofs for the complete generation pipeline, and \textbf{extension to alternative quantum-inspired distributions including Binomial and Negative Binomial statistics}. Furthermore, the method is well-suited for hardware acceleration through FPGA or ASIC implementation, and systematic benchmarking against existing classical and quantum RNGs will further establish its performance and security characteristics.

In summary, the digital coherent QRNG successfully bridges quantum-inspired statistical principles with practical classical computation, providing a verifiable, flexible, and accessible approach to high-quality random number generation. It demonstrates that essential statistical characteristics of quantum randomness can be realized through carefully designed computational processes, opening new opportunities for secure and scalable cryptographic applications while maintaining the mathematical rigor and certification capabilities essential for modern security systems.

\section*{Acknowledgments}

The author acknowledges valuable discussions with Michael Redding, CTO of Quantropi, whose encouragement was instrumental in exploring the information-theoretic and computational security foundations of this work. The author also acknowledges the use of AI tools including DeepSeek, Gemini, and ChatGPT for language refinement and editorial assistance during the preparation of this manuscript.

\section*{DATA AVAILABILITY STATEMENT}
All data that support the findings of this study are included within the article.

\bibliography{my}

\end{document}